\documentclass{paper}
\usepackage{indentfirst}
\usepackage{xcolor}
\usepackage[T1]{fontenc}

\usepackage{amssymb}
\usepackage{bm}
\usepackage{booktabs}
\usepackage{algorithmic}
\usepackage{algorithm}
\usepackage{amsmath,amsfonts}
\usepackage{array}
\usepackage{booktabs}
\usepackage{multirow}
\usepackage{stfloats}
\usepackage{graphicx}
\usepackage{epstopdf}
\usepackage{float}
\usepackage[caption=false,font=footnotesize,labelfont=rm,textfont=rm]{subfig}
\usepackage{balance}
\usepackage{verbatim}
\usepackage{pifont}
\usepackage{cite}

\newtheorem{definition}{Definition}
\newtheorem{theorem}{Theorem}
\newtheorem{proof}{Proof}

\title{Constructions of Control Sequence Set for Hierarchical Access in Data Link Network}
\author{Xianhua Niu, Jiabei Ma, Enzhi Zhou, Yaoxuan Wang, Bosen Zeng, and Zhiping Li}

\begin{document}
	\maketitle
	
	\begin{abstract}
		Time slots are a valuable channel resource in the data link network with time division multiple access architecture. The need for finding a secure and efficient way to meet the requirements of large access capacity, differentiated access, maximum utilization of time slot resource and strong anti-eavesdropping ability in data link networks is well motivated.
		In this paper, a control sequence-based hierarchical access control scheme is proposed, which not only achieves differentiated time slots allocation for the different needs and levels of nodes, but also enhances randomness and anti-interception performance  in data link networks.
		Based on the scheme, a new theoretical bound is derived to characterize parameter relationships for designing optimal hierarchical control sequence(HCS) set. Moreover, two flexible classes of optimal hierarchical control sequence sets are constructed.
		By our construction, the terminal user in the data link can access hierarchically and randomly and transmit data packets during its own hopping time slots of the successive frames to prevent eavesdropping while maintaining high throughput.
	\end{abstract}

	\section{Introduction}\label{intro}
	High capacity, high data transmission rate, and strong anti-eavesdropping capability play a crucial role in ensuring communication quality\cite{Sturdy}. Specifically, as a wireless communication system, the well-known data links like Link4, Link11, Link16\cite{C.H.Kao}, Link22\cite{R.L}, and the data link DTS-03 developed by China, all of which possess the aforementioned characteristics. Moreover, future research will continue to focus on achieving differentiated access communication.
	The issue of differentiated access communication refers to the varying communication effects that occur in practical application scenarios due to factors such as different nodes, different communication channels, and different channel conditions, even under a unified communication or access rule.
	For example, the bandwidth levels and transmission rates of different communication channels to some extent influence the duration of transmitting messages of varying lengths. Furthermore, there are varying interception risks associated with the transmission of messages in channels of different security levels.
	In the literature, several methods have been proposed to improve communication quality for differentiated access communication.
	For instance, in \cite{zhang}, Zhang, $\emph{et al}$. proposed an Opportunistic User Selection (OUS) scheme to address the problem of different user transmission rates in a bidirectional relay channel. In this scheme, users with better instantaneous signal-to-noise ratio in the end-to-end channel are given priority for access and transmission.
	Block, $\emph{et al}$.\cite{block} proposed an adaptive transmission protocol that provides corresponding channel state information before each data transmission. The protocol utilizes the received channel information to adapt the transmission parameters in direct sequence spread spectrum wireless networks, thereby improving system performance during data transmission, such as throughput efficiency.
	And in \cite{zhangY}, Zhang, $\emph{et al}$. proposed using a Layered Division Multiplexing (LDM) based hierarchical multicast scheme to improve network performance, such as throughput. The approach involves assigning each user corresponding to the broadcast content to the appropriate LDM layer and finding the optimal transmit power and data rate for each layer.
	We can see from those papers that it is also extremely important to achieve differentiated access communication in the data link system.
	\subsection{Access in Data Link}
	In wireless communication, achieving secure and interference-free access and data transmission between nodes has always been an important issue. Previous studies have also proposed various solutions to address the challenges of node access and transmission process encountering interferences.
	For example, complementary sequences have been proposed in the code division multiple access (CDMA) architecture to separate and reduce signal interference between communication nodes\cite{zhou,shen,Adhikary}.
	Furthermore, a channel hopping sequence is proposed to ensure that communication nodes converge at least once on a common channel within a specific period of time\cite{wang}.
	There are two main types of data link access strategies based on time division multiple access (TDMA) architecture as follows: "contention-based" channel access schemes and "collision avoidance" channel access schemes.
	The contention-based channel access scheme primarily use the techniques of carrier sensing, backoff waiting, and error frames retransmission, nodes that need to access the channel compete for access to the channel through a backoff algorithm. In recent years, the main studies on contention-based channel access schemes incldue ALOHA, Carrier Sense Multiple Access (CSMA), Floor Acquisition Multiple Access (FAMA), Dual Busy Tone Multiple Access (DBTMA)\cite{I.Demirkol,H.H.Choi1,W.Hu,A.N.Alvi,H.H.Choi2}, etc.
	The "collision avoidance" channel access schemes include static, dynamic, and sequence-based channel access techniques.
	The static access scheme uses a centralized time slot allocation method, which assigns fixed time slot resources to each network node\cite{X.Liu}.
	The dynamic channel access scheme dynamically calculates the number of time slots required by each node, using a distributed time slot allocation method to enable each node to obtain the necessary time slot resources, reducing end-to-end message delay, and improving channel utilization. The research findings of various types of dynamic channel access schemes have been publicly published for the past few years\cite{Z.G,A.M.Lewis,T.Zhang,F.Lyu}.
	The sequence-based channel access techniques include protocol sequences, conflict-avoiding codes, and hopping sequences.
	Protocol sequence is a periodic binary sequence designed to control node access to the channel for a multi-node access model on a collision channel without feedback\cite{J.Massey1,J.Massey2}. Each node sends data according to the assigned protocol sequence, ensuring that the data is successfully transmitted within a fixed delay and providing good stability for short-term communication performance. Up to now, there have been a lot of constructions of protocol sequence, such as  Chinese Remainder Theorem(CRT) sequence, prime sequence, extended prime sequence\cite{K.W.Shum,Y.Zhang,L.Q.Gui,F.Liu}, etc.
	Conflict-avoiding code is used to ensure that each node can successfully send data at least once within a sequence cycle under time slot synchronization\cite{V.I.Levenshtein,H.Fu,M.Mishima}.
	Hopping sequence is $q$-ary pseudo-random sequence with good anti-interference performance. It is designed to ensure that within one period, any pair of transmitting and receiving devices can converge on a common channel for all delays\cite{H.Shao,Niu1,Niu2}.
	However, the aforementioned techniques have not effectively addressed the issue of collision-free node access due to their limitations, such as allowing for time slots conflicts, unstable communication links, and low slot resource utilization efficiency.
	Consequently, a new type of sequence, known as Control Sequence(CS), was proposed.
	CS is a kind of $q$-ray pseudo-random sequence with zero correlation, zero time delay and non-periodic properties oriented to the data link.
	In 2018, Liu, $\emph{et al}$.\cite{Liu} first proposed the control sequence to solve the access control problems in data link network.
	In 2020, Niu, $\emph{et al}$.\cite{Niu3} proposed an access application control model access control scheme for the situation where all terminal users' frames are synchronised on the data link, and constructed a new control sequence set based on the model.
	In 2022, Tan, $\emph{et al}$.\cite{Tan} proposed a self-organizing model access control scheme and constructed a new control sequence set based on the scheme.
	These CS-based schemes allow each communication node to access in data link once in a time frame without conflict.
	However, in actual applications, due to the differences in task content between nodes, the priority of each node's data packets will vary, which means that nodes with higher packet priority need to occupy more time slots in a frame to meet the demand for rapid and efficient real-time data transmission(i.e., the node has a higher access level).
	As shown in Fig.\ref{fig:datalink}, assuming a Unmanned Aerial Vehicle(UAV) cluster exercise is taking place, the information exchange between the UAV and the control center becomes more frequent. In other words, the UAV access the data link more frequently within one frame, indicating a higher access level.
	\begin{figure}[tb]
		\begin{center}
			\includegraphics[width=\columnwidth]{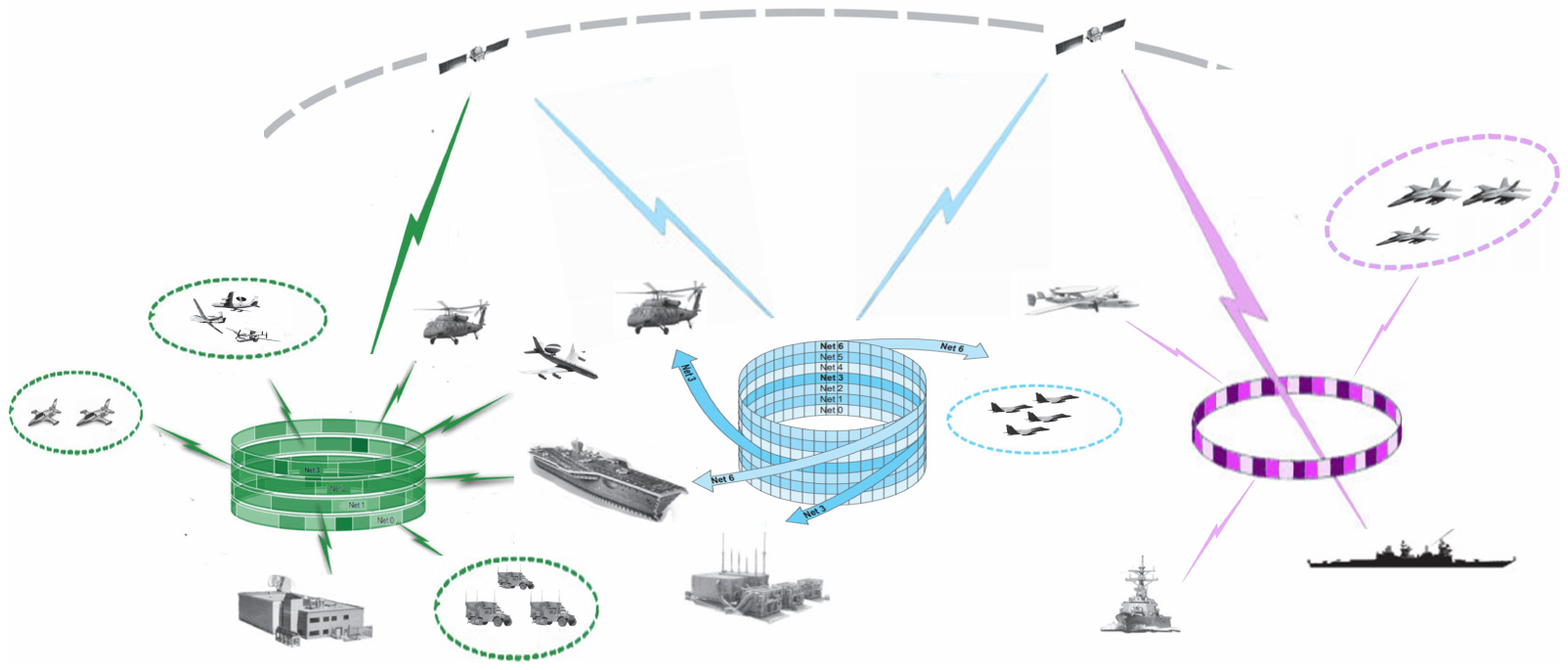}
		\end{center}
		\caption{Data Link Networks}
		\label{fig:datalink}
	\end{figure}
	Therefore, finding a solution to the problem of collision-free access between nodes under different conditions and states (i.e., the differentiated access communication problem) is what we need to focus.
	The rest of the paper is organized as follows: The main notations and definitions are introduced in Section~\ref{def}.
	Section~\ref{model} introduces the new hierarchical access model based on hierarchical control sequence and enumerates the conditions that hierarchical control sequences must satisfy.
	Section~\ref{Bound} presents the main theorem of our proposed new bound on the control sequence for hierarchical access.
	Two classes of HCS sets are constructed in Section~\ref{Cons}.
	Performance evaluation is given in Section~\ref{evaluate}. The last section concludes the paper.
	\section{Notations and Definitions}\label{def}
	For convenience, the main notations used in this paper are shown as follows:
	\begin{itemize}
		\item \emph{t}: The number of time slots in a time frame;
		\item \emph{$\lambda$}: The number of access levels;
		\item \emph{$r_i$}: The number of time that a user with level $i$ accesses the channel in a time frame, i.e., the level value, $0 \leq i \leq \lambda-1$;
		\item \emph{$u_i$}: The number of users with level $i$;
		\item $gcd(x,y)$: the greatest common divisor of $x$ and $y$;
		\item $\langle x \rangle_y$: $x$ modulo $y$;
	\end{itemize}
	Let $\mathcal{T}$ be a time slot set with size $|\mathcal{T}|=\emph{t}$, $\boldsymbol{S}$ be a set of $\textit{n}$ sequences of length $\textit{l}$ over $\mathcal{T}$.
	For any two sequences $\boldsymbol{x}=(x_0,x_1,...,x_{l-1})$, $\boldsymbol{y}=(y_0,y_1,...,y_{l-1})$ $\in \boldsymbol{S}$, the Hamming
	correlation function $\emph{H}_{\boldsymbol{x},\boldsymbol{y}}$ of $\boldsymbol{x}$ and $\boldsymbol{y}$ is defined as follows:
	\begin{equation}
		\emph{H}_{\boldsymbol{x},\boldsymbol{y}}(\tau)=\sum_{i=0}^{l-1}h(x_i,y_{i+\tau}),0\leq \tau < l
	\end{equation}
	where $h(a,b)=1$ if \textit{a}=\textit{b}, and $h(a,b)=0$ otherwise.
	As the new model is based on HCS, to distinguish between HCS and CS, we provide the definitions of HCS and CS.
	\begin{definition}
		A sequence set $\boldsymbol{S}$ is called a (\emph{l},\emph{M},\emph{t}) CS set over time slot set with size \emph{t} which includes \emph{M} sequences of length \emph{l}
		if the Hamming correlation $\emph{H}_{\boldsymbol{x},\boldsymbol{y}}(0)=0$ for any two control sequences $\boldsymbol{x}, \boldsymbol{y} \in \boldsymbol{S}$.
	\end{definition}
	Note that a CS set is applied for a group of users with the same access level.
	\begin{definition}
		A sequence set $\boldsymbol{S}$ is called a (\emph{l},\emph{M},\emph{t};\emph{$\sum_{i=0}^{\lambda-1}u_i$},\emph{$\lambda$}) HCS set with $\emph{M}$ sequences of length $\emph{l}$ over time slot set of size $\emph{t}$,
		which includes \emph{$\sum_{i=0}^{\lambda-1}u_i$} users with \emph{$\lambda$} access levels if the Hamming correlation $\emph{H}_{\boldsymbol{x},\boldsymbol{y}}(0)=0$ for any two sequences $\boldsymbol{x}, \boldsymbol{y} \in \boldsymbol{S}$.
	\end{definition}
	\section{New Hierarchical Access Control Model}\label{model}
	Based on the considerations of the aforementioned problem in Section~\ref{intro}, a Hierarchical Control Sequence(HCS)-based new hierarchical access control mode under TDMA architecture is proposed for nodes with different access levels.
	As shown in Fig.\ref{fig.newModel}, four users with three different access levels are given. User $\boldsymbol{A}$ and $\boldsymbol{B}$ have the same access level. User $\boldsymbol{B}$ is assigned to $\boldsymbol{S}_Y$ and transmits data from time slot $\boldsymbol{S}_Y(1)$ to $\boldsymbol{S}_Y(n-1)$, when the transmission completed, user $\boldsymbol{B}$ would release the $\boldsymbol{S}_Y$.
	User $\boldsymbol{A}$ is assigned to $\boldsymbol{S}_Y$ after completing system synchronization, and starts data transmission from time slot $\boldsymbol{S_Y}(n)$.
	The level value of user $\boldsymbol{C}$ is 2. After being assigned to $\boldsymbol{S_X}$, the user starts data transmission from time slot $\boldsymbol{S}_{X1}(2),\boldsymbol{S}_{X2}(2)$. The access process of user $\boldsymbol{D}$ is the same as above.
	We can also see from Fig.\ref{fig.newModel} that for the given number of time slots in a frame and access times of users, the number of users with each access level needs to be determined to ensure the maximum utilization of time slots. Therefore, a new theoretical bound on the access times of users, the number of users with different levels and the number of time slots is proposed.
	In addition, two classes of new HCS sets are constructed.
	\begin{figure}[htb]
		\begin{center}
			\includegraphics[width=\columnwidth]{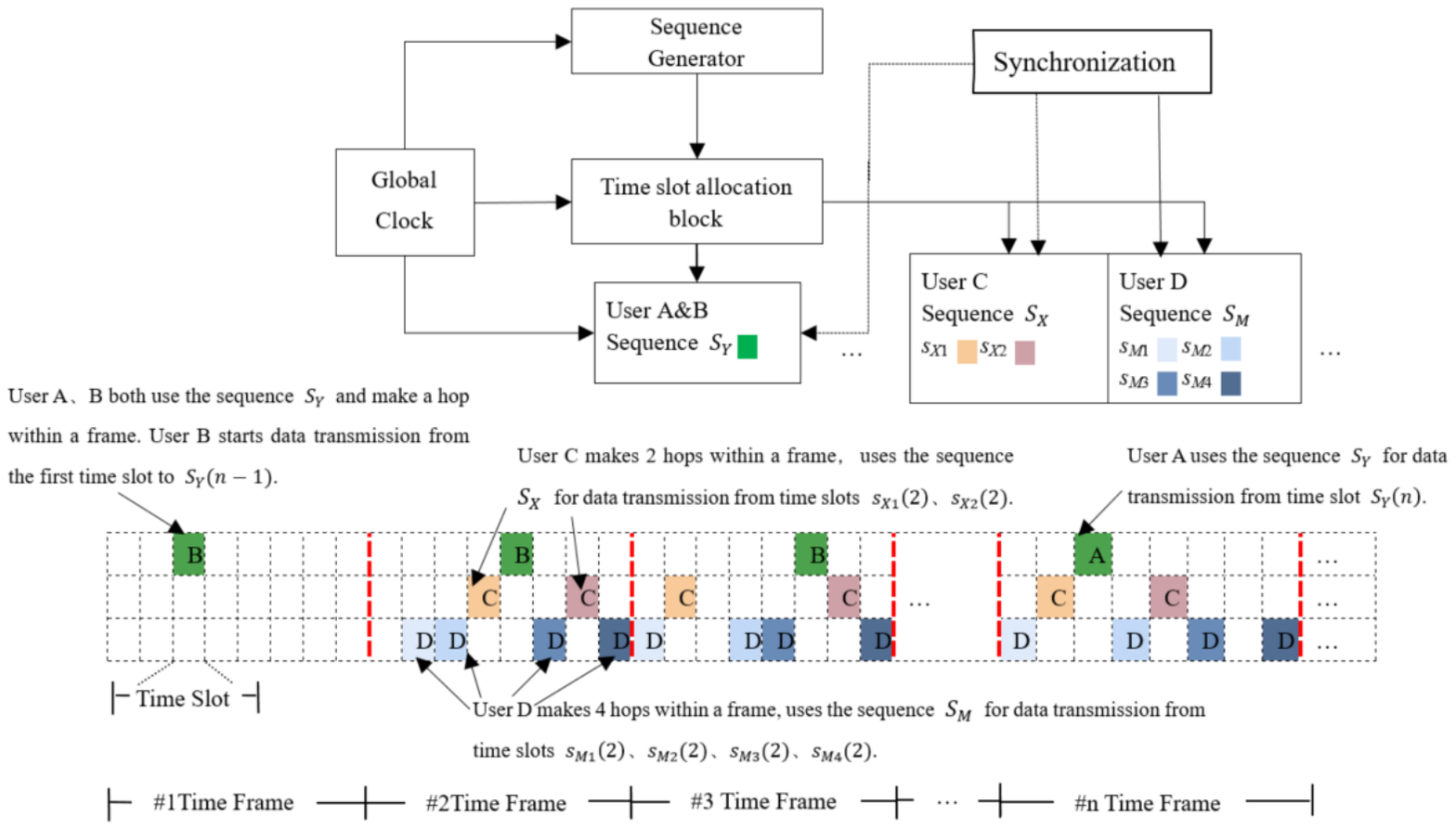}
		\end{center}
		\caption{Hierarchical Access Control Model.}
		\label{fig.newModel}
	\end{figure}
		To meet the demand of hierarchical access to the data link network, the HCSs should meet the following conditions:
	\begin{enumerate}
		\item[(1)] Every time slot in a frame can be used and the utilization of each time slot occurrences in a cycle would be equal;
		\item[(2)] The total number of time slots occupied by all users with different access levels should not be greater than the number of time slots in a time frame;
		\item[(3)] The length of sequence $\emph{l}$ should be large enough to ensure the randomness of the sequence;
		\item[(4)] The sequences should have enough randomness and complexity to improve the anti-interference ability and guarantee strictly conflict-free communication between nodes.
	\end{enumerate}
	\section{Theoretical Bound}\label{Bound}
	It is known in the research of frequency hopping sequence that the theoretical bound plays an important role in evaluating the performance of sequence sets.
	The sequence set is considered optimal only when its parameters meet the relevant theoretical bound.
	Therefore, we derived a new theoretical bound to better measure the excellence of HCS set in this section.
	From the definition of HCS set in section~\ref{def}, it can be seen that for a given number of time slots, user access level, and number of access levels, it is necessary to know the number of users with different access levels. So, we state the main theorem on the new bound on HCS.
	\begin{theorem}
		If there are $\emph{t}$ time slots in a time frame and $\lambda$ levels, for any level $i(0 \leq i \leq \lambda-1)$ and the number of users with level $i$, we have
		\begin{equation}
			\label{bound1}
			\begin{aligned}
				\sum^{\lambda-1}_{i=0}r_i\cdot u_i \leq \emph{t}
			\end{aligned}
		\end{equation}
	\end{theorem}
	\begin{proof}
		For any two levels $i, j(0 \leq i \neq j \leq \lambda-1)$, the corresponding level values in a frame are $r_i, r_j$, and the number of users are $u_i, u_j$.
		Let $d=gcd(\lfloor\frac{t}{r_i}\rfloor,\lfloor\frac{t}{r_j}\rfloor)$. So the number of users with level $i$ that can be accommodated is
		\begin{equation}
			\begin{aligned}
				\nonumber u_i &= \lfloor\frac{t}{r_id}\rfloor \cdot(d-\lceil \frac{u_ir_id}{t} \rceil)\\
				&\leq \frac{t}{r_jd}\cdot(d-\frac{u_ir_id}{t})\\
				&=\frac{t}{r_j}-\frac{u_ir_i}{r_j}
			\end{aligned}
		\end{equation}
		After the transfer of both sides of the equal sign, we have
		\begin{equation}\label{bound2}
			u_ir_i+u_jr_j \leq t
		\end{equation}
		Let $0 \leq i \leq \lambda-1$, we obtain a series of inequalities in form of $(\ref{bound2})$. Since index of each access levels occurs $(\lambda-1)$ times, we have
		\begin{equation}
			\label{bound3}
			(\lambda-1) \sum_{i=0}^{\lambda-1}r_i\cdot u_i \leq \emph{t}(\lambda-1)
		\end{equation}
		Then $(\ref{bound1})$ is proved from dividing both sides of $(\ref{bound3})$ by $(\lambda-1)$.
	\end{proof}
	There are two examples of this bound shown in Fig.\ref{fig:main}. It shows the feasible number and optimal number of users of three different access levels with respect to the new bound. All optimal numbers are on the plane of the upper bound.
	In Fig.\ref{a}, the values of the x,y, and z coordinates indicate the feasible and maximum number of users with three different levels.
	The points that fall on the plane denote that the value obtained by multiplying each of the three coordinate values of the point by the corresponding level values is equal to $\emph{t}$. These points are the optimal number of users with three different access levels.
	\begin{figure}[htb]
		\begin{center}
			\subfloat[\label{a} Access levels $r_0,r_1,r_2=1,2,6$]{\includegraphics[width=\columnwidth]{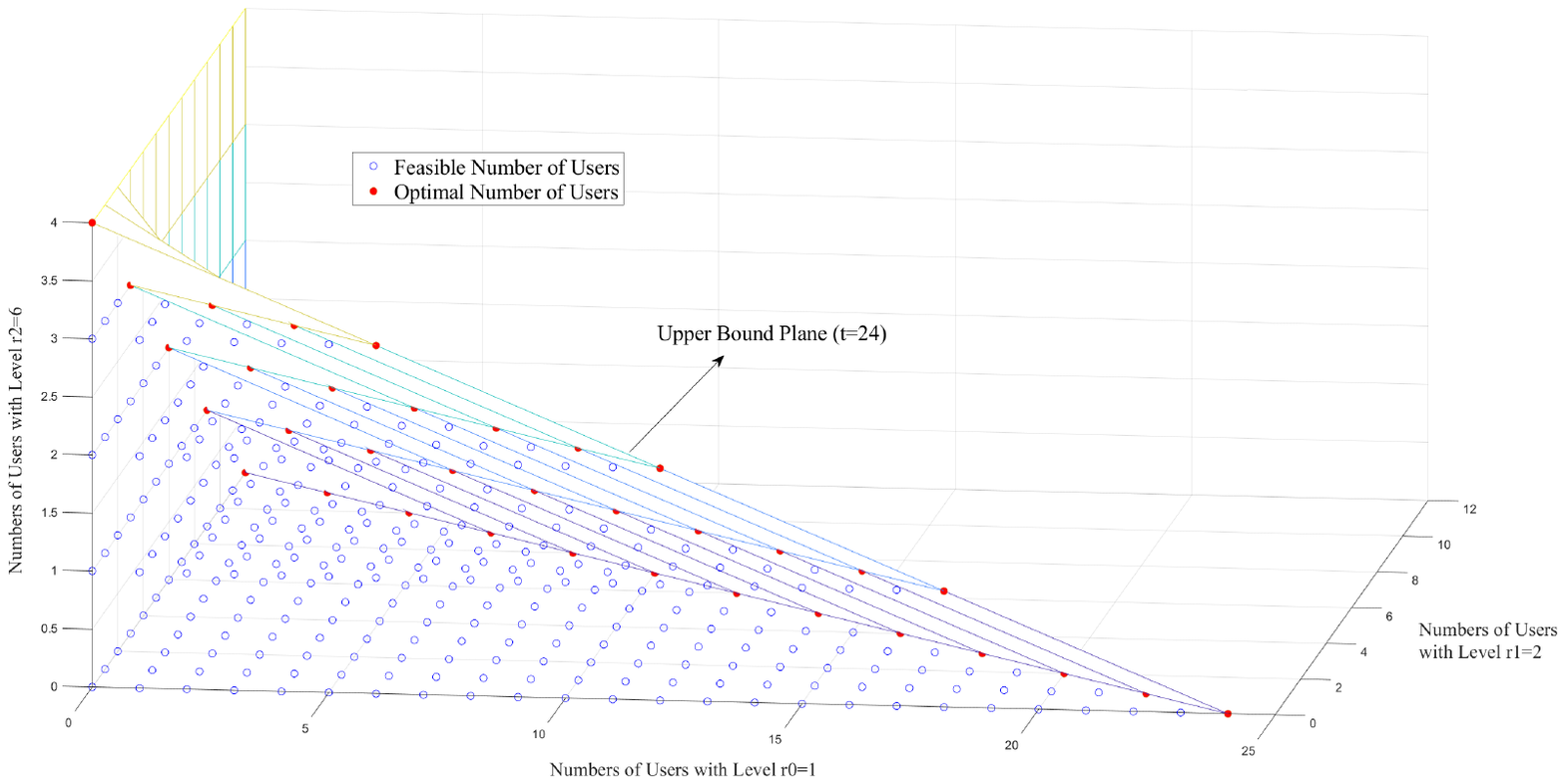}}
			\hfill
			\subfloat[\label{b} Access levels $r_0,r_1,r_2=3,5,15$]{\includegraphics[width=\columnwidth]{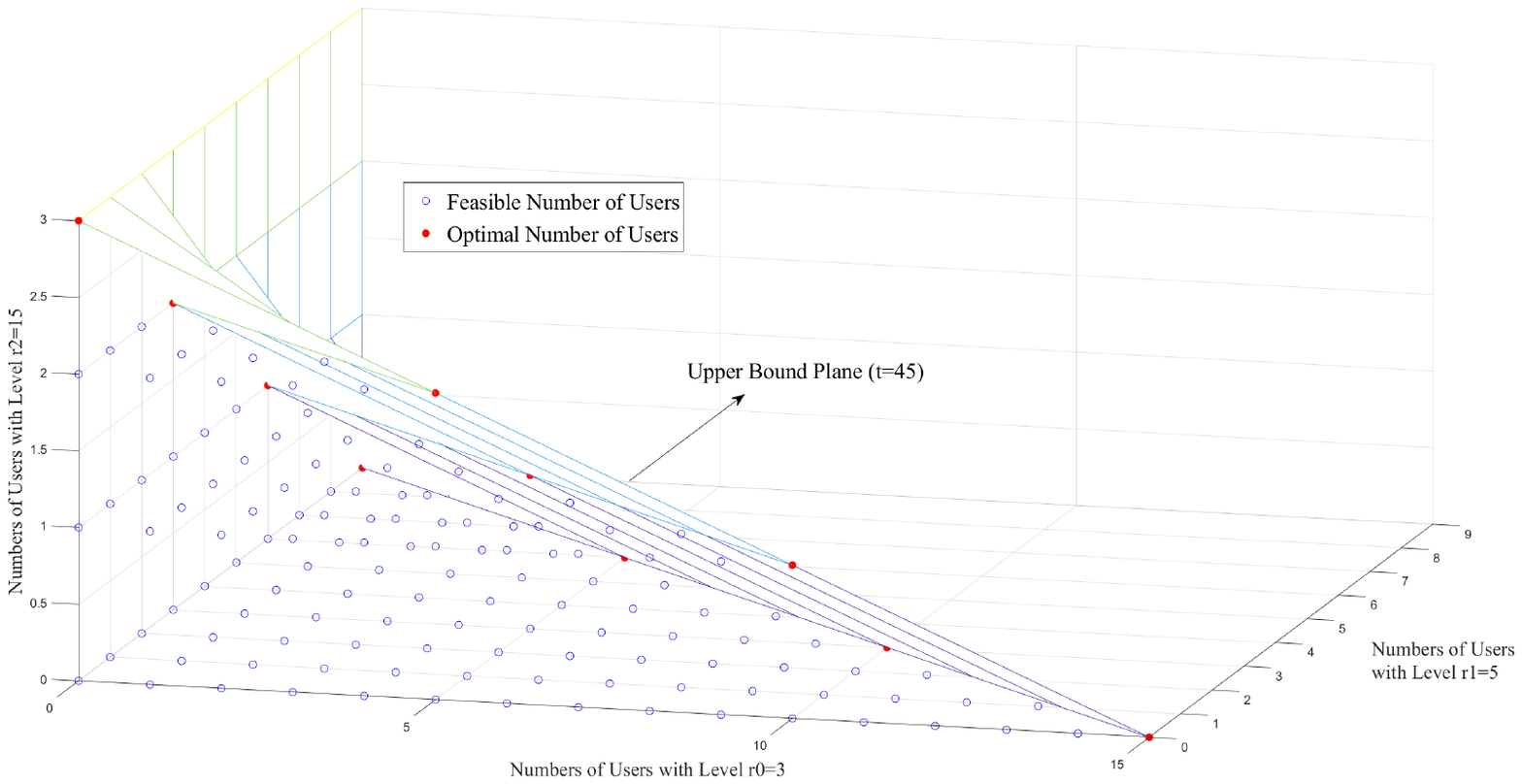}}
		\end{center}
		\caption{Feasible and optimal number of users with three access levels}
		\label{fig:main}
	\end{figure}
	\section{Constructions of HCS Set}\label{Cons}
	In this section, we present two constructions of hierarchical control sequence sets.
	These two constructions of HCS sets, while ensuring conflict-free user access to data links as implemented by existing control sequences, address the user differentiated access communication issue raised in Section~\ref{intro}.
	Let $\mathcal{T}=\{\emph{T}_0,\emph{T}_1,...$$,\emph{T}_{t-1}\}$ be a time slot set with size $|\mathcal{T}|=t$.
	\subsection{The First Construction}
	This construction is proposed based on a condition that multiplicative relationship between each access level and the number of time slots.
	Also, each level is multiplied by the maximum level.
	\emph{Step 1}: Let the level values are listed in ascending order(i.e., $r_0<r_1<r_2<...<r_{\lambda-1}$).
	Let $\emph{R}=r_{\lambda-1}$, $\mathcal{G}_m$ be an abelian group with size $\frac{t}{R}$,
	and the set $\boldsymbol{A}$ = \{ $\boldsymbol{a}^0,\boldsymbol{a}^1,...,\boldsymbol{a}^{\frac{t}{R}!-1}|\boldsymbol{a}^\gamma= \big(\textit{a}^\gamma(0),\textit{a}^\gamma(1),...,\textit{a}^\gamma(\frac{t}{R}-1) \big),0\leq\gamma\leq(\frac{t}{R}!-1) $\} consists of a total permutation of elemets of the group $\mathcal{G}_m$.
	Choose a pseudo-random sequence $\boldsymbol{b}^m=\big(\textit{b}^m(0),\textit{b}^m(1),...,\textit{b}^m(tR-1) \big)$ over the set $\{0,1,...,(\frac{t}{R}!-1)\}$.
	For example of $t=24$, $\lambda=3$, $r_0$ = 2, $u_0$ = 3, $r_1$ = 3, $u_1$ = 4, $r_2$ = $R$ = 6, $u_2$ = 1. For $\mathcal{G}_m$ = \{0,1,2,3\}, Let $\boldsymbol{A}$ = \{$\boldsymbol{a}^0, \boldsymbol{a}^1,..., \boldsymbol{a}^{23}$\}, and the total number of permutations are shown in Table~\ref{Permutation}.
	\begin{table}[t]
		\renewcommand{\arraystretch}{1.5}
		\begin{center}
			\caption{The total permutations of elemets over $\boldsymbol{A}$}
			\label{Permutation}
			\resizebox{\linewidth}{!}{
				\begin{tabular}{llllll}
					\hline
					$\boldsymbol{a}^0=(0,1,2,3)$   &   $\boldsymbol{a}^6=(1,0,2,3)$   &   $\boldsymbol{a}^{12}=(2,0,1,3)$   &   $\boldsymbol{a}^{18}=(3,0,1,2)$\\
					$\boldsymbol{a}^1=(1,0,2,3)$   &   $\boldsymbol{a}^7=(1,0,3,2)$   &   $\boldsymbol{a}^{13}=(2,0,3,1)$   &   $\boldsymbol{a}^{19}=(3,0,2,1)$\\
					$\boldsymbol{a}^2=(0,2,1,3)$   &   $\boldsymbol{a}^8=(1,2,0,3)$   &   $\boldsymbol{a}^{14}=(2,1,0,3)$   &   $\boldsymbol{a}^{20}=(3,1,0,2)$\\
					$\boldsymbol{a}^3=(0,2,3,1)$   &   $\boldsymbol{a}^9=(1,2,3,0)$   &   $\boldsymbol{a}^{15}=(2,1,3,0)$   &   $\boldsymbol{a}^{21}=(3,1,2,0)$\\
					$\boldsymbol{a}^4=(0,3,1,2)$   &   $\boldsymbol{a}^{10}=(1,3,0,2)$&   $\boldsymbol{a}^{16}=(2,3,0,1)$   &   $\boldsymbol{a}^{22}=(3,2,0,1)$\\
					$\boldsymbol{a}^5=(0,3,2,1)$   &   $\boldsymbol{a}^{11}=(1,3,2,0)$&   $\boldsymbol{a}^{17}=(2,3,1,0)$   &   $\boldsymbol{a}^{23}=(3,2,1,0)$\\
					\hline
				\end{tabular}
			}
		\end{center}
	\end{table}
	Choose a pseudo-random sequence $\boldsymbol{b}^m=(19,16,9,18,7,...,22,5)$.
	For the users with $r_0$ and $r_1$, choose the pseudo-random sequences over the abelian group $\mathcal{G}_0$ and $\mathcal{G}_1$ respectively.
	The pseudo-random sequences over $\mathcal{G}_0$:
	\[b^{0,0}=(1,2,1,1,0,...,2,0)\]
	\[b^{0,1}=(0,0,2,0,1,...,1,1)\]
	\[b^{0,2}=(2,1,0,2,2,...,0,2)\]
	The pseudo-random sequences over $\mathcal{G}_1$:
	\[b^{1,0}=(0,0,1,0,1,...,1,1)\]
	\[b^{1,1}=(1,1,0,1,0,...,0,0)\]
	\emph{Step 2}: For any user with level $i(0 \leq i \leq \lambda-1)$, we can generate a HCS, and then construct the HCS set with $\sum_{i=0}^{\lambda-1}u_i$ sequences of length $tR$ as $\boldsymbol{S}=\{\textbf{s}^i_j=\boldsymbol{s}^i_j(\alpha)|0 \leq i \leq \lambda-1, 0 \leq j \leq u_i-1, 0 \leq \alpha \leq tR-1 \}$:
	\begin{equation}
		\begin{aligned}
			\nonumber \boldsymbol{s}^i_j(\alpha) &= \big(z^i_{j,0}(\alpha),z^i_{j,1}(\alpha),...,z^i_{j,r_i-1}(\alpha)\big)\\
			&=\big(z^i_{j,\theta}(\alpha)|0 \leq \theta \leq r_i-1\big)
		\end{aligned}
	\end{equation}
	\begin{eqnarray}\label{ConsA1}
		z^i_{j,\theta}(\alpha) &=& \textit{a}^{b^m(\alpha)}\big(\langle \beta+\omega+ \lfloor \frac{j}{\eta_i} \rfloor \rangle_{\frac{t}{R}} \big) + \beta \cdot \frac{t}{R}
	\end{eqnarray}
	where $\eta_i=\frac{R}{r_i}$, $\beta=b^{i,\langle j \rangle_{\eta_i}}(\alpha)+\theta\cdot \eta_i$, $\omega=\frac{\sum^{i-1}_{\zeta=0}u_\zeta\cdot r_\zeta}{R}$. $\boldsymbol{b}^{i,0}=\big(b^{i,0}(0),b^{i,0}(1),...,b^{i,0}(tR-1)\big)$ is a pseudo-random sequences over an abelian group $\mathcal{G}_i$ with size $\eta_i$, and $\boldsymbol{b}^{i,\varepsilon}=\langle \boldsymbol{b}^{i,0}+1 \rangle_{\eta_i}, \varepsilon \leq \eta_i-1$.
	For the example in Step 1, applying (\ref{ConsA1}), the (144,8,24;8,3) HCS set $\boldsymbol{S}$ is obtained and shown in Table~\ref{HCS set1}.
	\begin{table*}[t]
		\renewcommand{\arraystretch}{1.5}
		\begin{center}
			\caption{The hierarchical control sequence set1}
			\label{HCS set1}
			\begin{tabular}{l}
				\hline
				$\textbf{s}^0_0=((4,19),(8,23),(6,17),...,(8,22),(0,13))$\\
				$\textbf{s}^0_1=((3,13),(2,13),(11,22),...,(6,19),(7,16))$\\
				$\textbf{s}^0_2=((10,20),(7,18),(1,12),...,(3,13),(10,23))$\\
				\hline
				$\textbf{s}^1_0=((0,9,16),(3,9,19),(7,13,23),...(4,15,20),(6,12,22))$\\
				$\textbf{s}^1_1=((6,15,22),(4,14,20),(2,8,18),...,(2,9,18),(3,9,19))$\\
				$\textbf{s}^1_2=((2,11,18),(0,10,16),(4,14,20),...,(5,14,21),(5,15,21))$\\
				$\textbf{s}^1_3=((5,12,21),(5,15,21),(3,9,19),...,(0,11,16),(2,8,18))$\\
				\hline
				$\textbf{s}^2_0=((1,7,8,14,17,23),(1,6,11,12,17,22),(0,5,10,15,16,21),...,(1,7,10,12,17,23),(1,4,11,14,17,20))$\\
				\hline
			\end{tabular}
		\end{center}
	\end{table*}
	\begin{theorem}
		The sequence set $\boldsymbol{S}$ with parameters $(tR,\sum_{i=0}^{\lambda-1}u_i,t;\sum_{i=0}^{\lambda-1}u_i,\lambda)$ obtained from above construction is a HCS set and each time slot $\emph{T}_\tau$ $(\emph{T}_\tau \in \mathcal{T},0\leq \tau \leq t-1)$ appears $\emph{t}R$ times in $\boldsymbol{S}$.
	\end{theorem}
	\begin{proof}
		It can be obtained from the above construction that the sequence set $\boldsymbol{S}$ is with $t$ sequences of length $tR$ over a time slot set of size $t$.
		Then we will prove the Hamming correlation $H(\boldsymbol{S})=0$. For any two sequences $\textbf{s}_j^i,\textbf{s}_{j'}^{i'} \in \boldsymbol{S}$, the Hamming correlation can be expressed as follows:
		\[H_{\textbf{s}^i_j,\textbf{s}_{j'}^{i'}}(0)=\sum^{tR-1}_{l=0}h\big(\textbf{s}_j^i(l),\textbf{s}_{j'}^{i'}(l)\big)\]
		We can prove the above equation by dividing it into the following three cases:
		\begin{enumerate}
			\item[(1)] $\textbf{s}_j^i \neq \textbf{s}_{j'}^i$ for any $j_i \neq j'$;
			\item[(2)] $\textbf{s}_j^i \neq \textbf{s}_j^{i'}$ for any $i \neq i'$;
			\item[(3)] $\textbf{s}_j^i \neq \textbf{s}_{j'}^{i'}$ for any $i \neq i', j \neq j'$.
		\end{enumerate}
		According to $($\ref{ConsA1}$)$, we have
		\begin{equation}
			\begin{aligned}
				\nonumber z^i_{j,\theta}(\alpha) =& \textit{a}^{b^m(\alpha)}\big(\beta+\omega+\lfloor \frac{j}{\eta_i} \rfloor \rangle_{\frac{t}{R}} \big)+\beta \cdot \frac{t}{R}
			\end{aligned}
		\end{equation}
		\begin{equation}
			\begin{aligned}
				\nonumber z^i_{j',\theta'}(\alpha) =& \textit{a}^{b^m(\alpha)}\big(\beta+\omega+\lfloor \frac{j'}{\eta_i} \rfloor \rangle_{\frac{t}{R}} \big)+\beta \cdot \frac{t}{R}
			\end{aligned}
		\end{equation}
		Since $j \neq j'$, whether $\theta=\theta'$ or not, we have
		\begin{equation}
			\begin{aligned}
				\nonumber &b^{i,\langle j \rangle_{\eta_i}}(\alpha)+\theta \cdot \eta_i\\
				&\neq b^{i,\langle j' \rangle_{\eta_i}}(\alpha)+\theta \cdot \eta_i
			\end{aligned}
		\end{equation}
		So we can obtain that
		\[z^i_{j,\theta}(\alpha)\neq z^i_{j',\theta'}(\alpha)\]
		It is obvious that:
		\[\textbf{s}^i_j \neq \textbf{s}^{i'}_{j'}\]
		for any $i,i'(0 \leq i \neq i' \leq \lambda-1)$.
		For the second and third case, we can prove them according to the equation$($\ref{ConsA1}$)$ similarly.
		So we get the Hamming correlation $H(\boldsymbol{S})=0$.\\
		Thus, we can obtain that for any two subsequences $\boldsymbol{z}^i_{j,\theta},\boldsymbol{z}^{i'}_{j',\theta'}$, we have
		\[z^i_{j,\theta}(\alpha)\neq z^{i'}_{j',\theta'}(\alpha)\]
		There are $t$ subsequences in $\boldsymbol{S}$, so each time slot $\emph{T}_\tau$ $(\emph{T}_\tau \in \mathcal{T},0\leq \tau \leq t-1)$ occurs once respectively in the corresponding position of each of these $t$ subsequences.
		Since the length of these subsequences is $tR$, we can obtain that each time slot $\emph{T}_\tau$ $(\emph{T}_\tau \in \mathcal{T},0\leq \tau \leq t-1)$ appears $tR$ times in $\boldsymbol{S}$.
	\end{proof}
	\begin{theorem}
		The HCS set $\boldsymbol{S}$ obtained from the above construction is optimal with respect to the new bound $(\ref{bound1})$.
	\end{theorem}
	\begin{proof}
		For the given access levels and the number of time slots, if a set of the number of users such that the number of time slots occupied by all users in a frame is equal to $t$, then it is optimal on the theoretical bound with respect to the number of users.
		Assume that there are $\lambda$ access levels, the maximum level value is $R$ and the corresponding number of users is $u_R$. Since $R$ and the number of time slots $t$ are multiplicative $(i.e.R|t)$, so the maximum number of users with $R$ is equal to $\frac{t}{R}$.
		Meanwhile, any level $i(0 \leq i \leq \lambda-1)$ and $R$ are multiplicative, so we can find the number $u_i$ of users with any level $i$ such that
		\[\sum_{i=0}^{\lambda-1}\frac{u_i}{\eta_i}=\frac{t}{R}\]
		by $\eta_i=\frac{R}{r_i}$, we can obtain that
		\[\sum_{i=0}^{\lambda-1}u_i\cdot r_i=t\]
		It is clearly the hierarchical control sequence set obtained from above construction is optimal with respect to the new theoretical bound.
	\end{proof}
	\begin{algorithm}[h]
		\caption{The First Construction}
		\label{alg:A}
		\begin{algorithmic}[1]
			\REQUIRE the number of time slots $t$, the number of access levels $\lambda$, the level values $r_i(0\leq i \leq \lambda)$, the number of users $u_i$ with level $i$, $\boldsymbol{A}=\{\boldsymbol{a}^0,\boldsymbol{a}^1,...,$ $\boldsymbol{a}^{\frac{t}{R}!-1}|\boldsymbol{a}^\gamma= \big(\textit{a}^\gamma(0),\textit{a}^\gamma(1),...,\textit{a}^\gamma(\frac{t}{R}-1) \big),0\leq\gamma\leq(\frac{t}{R}!-1)\}, \boldsymbol{b}^m=\big(\textit{b}^m(0),\textit{b}^m(1),...,\textit{b}^m(tR-1) \big),\boldsymbol{b}^{i,\varepsilon}=\langle \boldsymbol{b}^{i,0}+1 \rangle_{\eta_i}, \varepsilon \leq \eta_i-1$
			\ENSURE a $(tR,\sum_{i=0}^{\lambda-1}u_i,t;\sum_{i=0}^{\lambda-1}u_i,\lambda)$ HCS set $\boldsymbol{S}$
			\STATE Initial $R=max\{r_i\}$,$\eta_i=\frac{R}{r_i}$, $\omega=\frac{\sum^{i-1}_{\zeta=0}u_\zeta\cdot r_\zeta}{R}$, $\beta=b^{i,\langle j \rangle_{\eta_i}}(\alpha)+\theta\cdot \eta_i$, $\textbf{z}^i_{j,\theta}=\phi$, $\boldsymbol{S}=\{\textbf{s}^i_j\}=\phi$;
			\FOR{$\alpha=0$ to $tR-1$}
			\FOR{$\theta=0$ to $r_i-1$}
			\STATE $z^i_{j,\theta}(\alpha)=\beta\cdot \frac{t}{R}+a^{b^m(\alpha)}\big(\langle \beta+\omega+ \lfloor \frac{j}{\eta_i} \rfloor \rangle_{\frac{t}{R}} \big);$
			\IF{$\alpha==tR-1$}
			\STATE $\textbf{s}^i_j=\textbf{z}^i_{j,\theta};$
			\STATE Update $\boldsymbol{S}$;
			\ENDIF
			\ENDFOR
			\ENDFOR
		\end{algorithmic}
	\end{algorithm}
	\subsection{The Second Construction}
	\emph{Step 1}: Define the initial state of the control sequence set as follows:
	\[\boldsymbol{C}^0=\Big\{\boldsymbol{c}^0_k=\big(\textit{c}^0_k(0),\textit{c}^0_k(1),...,\textit{c}^0_k(t-1)\big):0\leq k \leq t-1 \Big\}\]
	where $\textit{c}^0_k(b_0)=k+b_0$ mod \emph{t}, $0 \leq b_0 \leq t-1.$
	For example $t=8$, the number of access levels $\lambda=3$. Let $r_0=1, u_0=1; r_1=3, u_1=1; r_2=4, u_2=1$. By Step 1 we obtain the initial set of sequence set as follows:
	\[\boldsymbol{c}^0_0=(0,1,2,3,4,5,6,7)\]
	\[\boldsymbol{c}^0_1=(1,2,3,4,5,6,7,0)\]
	\[\boldsymbol{c}^0_2=(2,3,4,5,6,7,0,1)\]
	\[\boldsymbol{c}^0_3=(3,4,5,6,7,0,1,2)\]
	\[\boldsymbol{c}^0_4=(4,5,6,7,0,1,2,3)\]
	\[\boldsymbol{c}^0_5=(5,6,7,0,1,2,3,4)\]
	\[\boldsymbol{c}^0_6=(6,7,0,1,2,3,4,5)\]
	\[\boldsymbol{c}^0_7=(7,0,1,2,3,4,5,6)\]
	\emph{Step 2}: Choose \emph{g} $\in \mathbb{Z}^*_t$ such that \emph{g} has the largest multiplicative order in $\mathbb{Z}^*_t$. Assume that $ord(g)=d$, note that \emph{d} is upper bounded by $\phi(t)$, where $\phi$ is the Euler function. In general, we have $\emph{d}=\Omega(t)$.
	For the example in Step 1, choose $g=3\in \mathbb{Z}^*_8$ such that $g$ has the largest multiplicative order $ord(g)=d=\phi(8)=4 \in \mathbb{Z}^*_8$.
	\emph{Step 3}: For any user with level $i(0 \leq i \leq \lambda-1)$, we can generate a HCS by \emph{n}-rounds of iterative operation, and then construct the desired HCS set with $\boldsymbol{S}=\{\textbf{s}_j^i=\boldsymbol{s}_j^i(b_n)|0 \leq i \leq \lambda-1, 0 \leq j \leq u_i-1 \}$:
	\begin{equation}
		\begin{aligned}
			\nonumber \boldsymbol{s}_j^i(b_n) &= \big(z^i_{j,0}(b_n),z^i_{j,1}(b_n),...,z^i_{j,r_i-1}(b_n)\big)\\
			&=\big(z^i_{j,\theta}(b_n)|0 \leq \theta \leq r_i-1\big)
		\end{aligned}
	\end{equation}
	\begin{equation}\label{rule}
		\begin{aligned}
			z^i_{j,\theta}(b_n) &= c^n_{\omega+j\cdot r_i+\theta}(b_n)\\
			&=\textit{g}^{a_n}\textit{g}^{a_{n-1}}\cdot \cdot \cdot\textit{g}^{a_1}\textit{c}^0_{\langle \omega+j\cdot r_i+\theta+a_1 \rangle_t}(b_0)~~mod~~\textit{t}
		\end{aligned}
	\end{equation}
	where $\omega=\frac{\sum^{i-1}_{\zeta=0}u_\zeta\cdot r_\zeta}{R}$, $b_n=a_n\cdot d^{n-1}t+b_{n-1}$, $0\leq a_n \leq d-1, 0\leq b_n \leq d^nt-1$.
	For the example in Step 2, choose $\emph{n}=2$, construct the control sequence set $\boldsymbol{C}^2=\{\boldsymbol{c}^2_k:0\leq k \leq 7\}$ by 2-round iterations,
	\[\boldsymbol{c}^2_k=(3^0\boldsymbol{c}^1_k,3^1\boldsymbol{c}^1_k,3^2\boldsymbol{c}^1_k,3^3\boldsymbol{c}^1_k)\]
	\[=(3^{\langle 0+0 \rangle_4}\boldsymbol{c}^0_k,3^{\langle 0+1 \rangle_4}\boldsymbol{c}_{\langle k+1 \rangle_8}^0,3^{\langle 0+2 \rangle_4}\boldsymbol{c}_{\langle k+2 \rangle_8}^0,\dots\]
	\[3^{\langle 3+2 \rangle_4}\boldsymbol{c}_{\langle k+2 \rangle_8}^0,3^{\langle 3+3 \rangle_4}\boldsymbol{c}_{\langle k+3 \rangle_8}^0)\]
	According to $($\ref{rule}$)$, construct the desired $(128,3,8;3,3)$ HCS set $\boldsymbol{S}$, shown in Table~\ref{HCS set2}.
	\begin{table*}[t]
		\renewcommand{\arraystretch}{1.5}
		\begin{center}
			\caption{The hierarchical control sequence set2}
			\label{HCS set2}
			\begin{tabular}{l}
				\hline
				$\textbf{s}^0_0=(0,1,2,3\dots,1,4,7,\dots,0,1,2)$\\
				\hline
				$\textbf{s}_0^1=((1,2,3),(2,3,4),(3,4,5),(4,5,6),\dots,(4,7,2),(7,2,5),(2,5,0),\dots,(1,2,3),(2,3,4),(3,4,5))$\\
				\hline
				$\textbf{s}_0^2=((4,5,6,7),(5,6,7,0),(6,7,0,1),(7,0,1,2),\dots,(5,0,3,6),(0,3,6,1),(3,6,1,4),\dots(4,5,6,7),(5,6,7,0),(6,7,0,1))$\\
				\hline
			\end{tabular}
		\end{center}
	\end{table*}
	\begin{theorem}\label{timeB}
		Each time slot $\emph{T}_\tau$ $(\emph{T}_\tau \in \mathcal{T},0 \leq \tau \leq t-1)$ appears $\emph{d}^n$ times in $\boldsymbol{S}$.
	\end{theorem}
	\begin{proof}
		Let $\alpha$ be an element of $\mathbb{Z}_t$. We count the number of indices $b_n$ such that $c^n_k(b_n)=\alpha$, i.e. $\emph{g}^{a_n}\emph{g}^{a_{n-1}}\cdots \emph{g}^{a_1}(\langle k+a_1\rangle_t+b_0)\equiv \alpha$ mod \emph{t},
		where $0\leq a_n \leq d-1, 0\leq b_n \leq d^nt-1$. Thus, we have $b_0\equiv \alpha\cdot g^{\langle -(a_n+a_{n-1}+\cdots +a_1) \rangle_d}$ mod \emph{t}. For each $0\leq a_n \leq d-1$, there exists a unique $b_0$ with $0\leq b_0 \leq t-1$ such that
		$b_0\equiv \alpha\cdot g^{\langle -(a_n+a_{n-1}+\cdots +a_1) \rangle_d}-(k+a_1)$ mod \emph{t}, it implies that there are exactly $p^n$ indices $b_n$ such that $z_{j,k}^i(b_n)=\alpha$.
	\end{proof}
	\begin{theorem}
		The sequence set $\boldsymbol{S}$ with parameters $(d^nt,\sum_{i=0}^{\lambda-1}u_i,t;\sum_{i=0}^{\lambda-1}u_i,\lambda)$ obtained from above construction is an optimal HCS set with respect to the new bound.
	\end{theorem}
	\begin{proof}
		It can be obtained from the above construction that the control sequence set $\boldsymbol{C}^n$ is with \emph{t} sequences of length $\emph{d}^nt$ over a time slot set of size \emph{t}.
		For any two sequences $\boldsymbol{c}^n_{k_1},\boldsymbol{c}^n_{k_2} \in \boldsymbol{C}^n$, suppose $\emph{H}(\boldsymbol{c}^n_{k_1},\boldsymbol{c}^n_{k_2})\neq 0$,
		so there exists any $0\leq k_1 \neq k_2 \leq t-1$, and $0\leq b_n \leq d^nt-1$ satisfying $\emph{c}^n_{\langle k_1 + a_1 \rangle_t}(b_n)=\emph{c}^n_{\langle k_2 + a_1 \rangle_t}(b_n)$ mod \emph{t}.
		So, we have
		\begin{center}
			\resizebox*{\columnwidth}{!}{
				$\emph{g}^{a_n}\emph{g}^{a_{n-1}}\cdots \emph{g}^{a_1}(k_1+a_1+b_0)\equiv \emph{g}^{a_n}\emph{g}^{a_{n-1}}\cdots \emph{g}^{a_1}(k_2+a_1+b_0)$ mod \emph{t}.
			}
		\end{center}
		This gives
		\begin{center}
			$\emph{g}^{\langle a_n+a_{n-1}+\cdots +a_1 \rangle_d}(k_1-k_2)\equiv 0$ mod \emph{t}.
		\end{center}
		Then we obtain $k_1 \equiv k_2$ mod \emph{t} for $\emph{g}^{\langle a_n+a_{n-1}+\cdots +a_1\rangle_d} \in \mathbb{Z}^*_t$. A contradiction. So we get the Hamming correlation $H(\boldsymbol{S})=0$.
		By step 1 of the construction, we can obtain that there are $t$ sequences in initial sequence set $\boldsymbol{C}$, it is obvious that every time slot in a time frame will be occupied.
		So the equal sign in the theoretical bound holds, i.e. the HCS set obtained from above construction is optimal with respect to the new bound.
	\end{proof}
	\begin{algorithm}[h]
		\caption{The Second Construction}
		\label{alg:B}
		\begin{algorithmic}[1]
			\REQUIRE integer $g \in \mathbb{Z}^*_t$, integer $n$, the number of time slots $t$, the number of access levels $\lambda$, the level values $r_i(0\leq i\leq \lambda-1)$, the number of users $u_i$ with level $i$, initial set $\boldsymbol{C}^0=\{\textbf{c}^0_k|0\leq k \leq t-1\}$
			\ENSURE a $(d^nt,\sum_{i=0}^{\lambda-1}u_i,t;\sum_{i=0}^{\lambda-1}u_i,\lambda)$ HCS set $\boldsymbol{S}$
			\STATE Initial $\textbf{c}^0_k=k+b_0$ $mod$ $t$, $0\leq b_0 \leq t-1$, $ord(g)=d=\varphi(t)$, $\omega=\frac{\sum^{i-1}_{\zeta=0}u_\zeta\cdot r_\zeta}{R}$, $\boldsymbol{S}=\{\textbf{s}_j^i\}=\phi$, $\textbf{z}^i_{j,\theta}=\phi$;
			\FOR{$b_n=0$ to $d^nt-1$}
			\FOR{$\theta=0$ to $r_i-1$}
			\STATE $c_{\omega+j\cdot r_i+\theta}^n(b_n)$
			\STATE $=g^{a_n}g^{a_{n-1}}\dots g^{a_1}c^0_{\langle \omega+j\cdot r_i+\theta+a_1 \rangle_t}(b_0)$ $mod$ $t$;
			\STATE $z_{j,\theta}^i=c_{\omega+j\cdot r_i+\theta}^n(b_n)$;
			\IF{$b_n==d^nt-1$}
			\STATE $\textbf{s}_j^i=\textbf{z}_{j,\theta}^i;$
			\STATE Update $\boldsymbol{S}$;
			\ENDIF
			\ENDFOR
			\ENDFOR
		\end{algorithmic}
	\end{algorithm}
	Since there can be different numbers of users satisfying the theoretical bounds in Section \ref{Bound} for the same set of time slots, number of access levels, and access level values, the above examples of the two constructions are only an optimization case.
	From the above two constructions, we can obtain that the first one emphasizes the time slots occupied by users within a frame, which are relatively evenly distributed, while the time slots occupied by users between frames are random.
	The second construction is that the time slots occupied by users are random both within a frame and between frames.
	\subsection{Access Control Protocol Based on HCS}
	\begin{figure}[htb]
		\begin{center}
			\includegraphics[width=\linewidth]{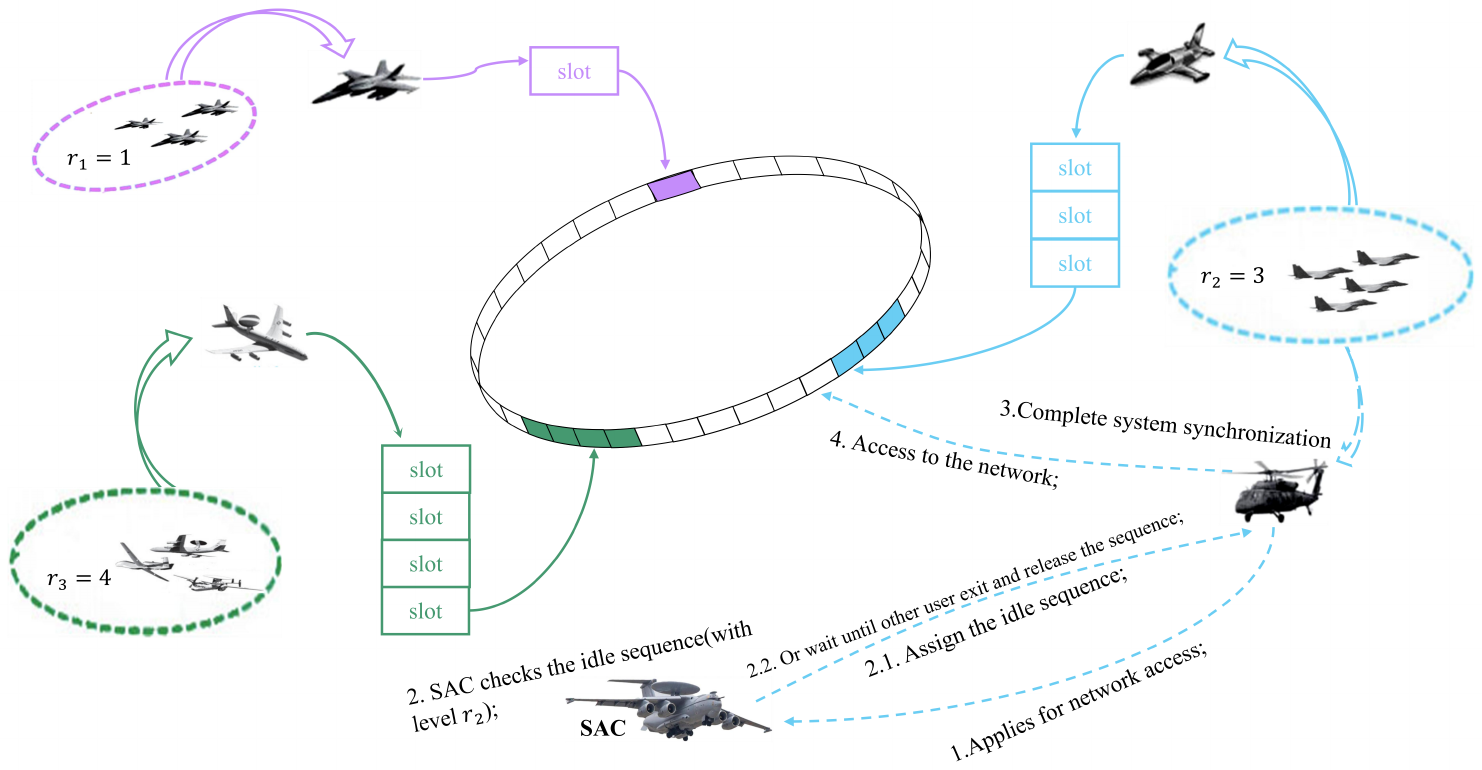}
		\end{center}
		\caption{Access to the group in data link. The solid line represents that the node has been entered, while the dotted line indicates that the node is applying for access in.}
		\label{fig:access process}
	\end{figure}
	In this part, we describe a rule of access control based on HCS generated by the above constructions.
	First of all, we define an HCS set $\boldsymbol{S}=\{\textbf{s}_j^i|0\leq i \leq \lambda-1,0\leq j \leq u_i-1\}$ to achieve the hierarchical access mechanism of data link.
	Once a user applies for network access, the Sequence Assignment Center(SAC) will randomly assign a sequence corresponding to the access level they belong to.
	We assume that the period of time is divided into $N$ time frames, and each of which is divided into $t$ time slots.
	Every element in $\boldsymbol{s}_j^i(\delta)=\big(z^i_{j,0}(b_n),z^i_{j,1}(b_n),...,z^i_{j,r_i-1}(b_n)\big)$ is used to determine the time slot position of the user to transmit data.
	Then, we will describe this rule from the process of entering the network.
	SAC initializes the HCS set as $\boldsymbol{S}=\{\boldsymbol{S}^i|0\leq i \leq \lambda-1\}$ and $|\boldsymbol{S}^i|=u_i$. And the specific process is shown in Fig.\ref{fig:access process}:
	\begin{itemize}
		\item A certain user with $r_2=3$ applies for network access at a random frame($j$-th frame);
		\item SAC checks if there is an idle sequence in $\boldsymbol{S}^2$;
		\item If there is an idle sequence, then it will be allocated to the user, who will start data transmission using the corresponding time slots in the sequence from $j$-th frame after completing system synchronization;
		\item If there is no available sequence, the user must wait until another user at the same level exits the network and releases sequence. SAC will then assign the reclaimed sequence to the user;
		\item After the sequence is allocated, the user will synchronize with the system, and then use the time slots corresponding to the $j$-th frame in the sequence for access.
	\end{itemize}
	\section{Performance Evaluation of HCS}\label{evaluate}
	In this section, we analyze and evaluate HCSs from two perspectives: communication applicability and anti-interference performance in data link networks with TDMA architecture.
	\subsection{Sequence Parameters Comparison}
	We compare the protocol sequences, frequency hopping sequences, and existing control sequences mentioned earlier in this paper with the HCSs presented in Section~\ref{Cons}. The comparison primarily focuses on the applicability in data link with TDMA architecture, suitability for collision-free communication among users, and suitability for multi-level access communication among users, which is shown in Table~\ref{table:comparison}.
	\begin{table*}[h]
		\renewcommand{\arraystretch}{1.5}
		\begin{center}
			\caption{Comparison of Parameters}
			\label{table:comparison}
			\resizebox*{\linewidth}{!}{
				\begin{tabular}{|c|c|c|c|c|c|c|c|}
					\hline
					\multirow{2}{*}{Ref}                      &     \multirow{2}{*}{Length}                    &     \multirow{2}{*}{Family Size}                                      &     Time Slot Size/                                       &     \multirow{2}{*}{Number of Users}                                        &     applicability                           &     collision-free                                                    &     \multirow{2}{*}{Hierarchical Access}\\
					&                                                &                                                                       &     Alphabet Size                                         &                                                                             &     in data link                            &     communication                                                     &      \\
					\hline
					Protocol Sequence\cite{L.Q.Gui}           &     $pq$                                       &     $p$                                                               &     /                                                     &     $K$                                                                     &     \ding{55}                               &     \ding{55}                                                         &     \ding{55}\\
					\hline
					Frequency Hopping Sequence\cite{Niu2}(Con.A)      &     $q^m-1$                            &     $q^{m-t}$                                                         &     $q^{m-t}$                                             &     $q^{m-t}$                                                               &     \ding{55}                               &     \ding{55}                                                         &     \ding{55}\\
					\hline
					Frequency Hopping Sequence\cite{Niu2}(Con.B)      &     $q^m-1$                            &     $\frac{q^{m-t}-1}{r}$                                             &     $\frac{q^{m-t}-1}{r}+1$                               &     $\frac{q^{m-t}-1}{r}$                                                   &     \ding{55}                               &     \ding{55}                                                         &     \ding{55}\\
					\hline
					Control Sequence\cite{Liu}(Con.A)         &     $l$                                        &    $n$                                                                &    $m$                                                    &    $n$                                                                      &     \ding{51}                               &     \ding{51}                                                         &     \ding{55}\\
					\hline
					Control Sequence\cite{Liu}(Con.B)         &     $l$                                        &    $n$                                                                &    $m$                                                    &    $n$                                                                      &     \ding{51}                               &     \ding{51}                                                         &     \ding{55}\\
					\hline
					Control Sequence\cite{Niu3}               &     $d^jm$                                     &    $m$                                                                &    $m$                                                    &    $m$                                                                      &     \ding{51}                               &     \ding{51}                                                         &     \ding{55}\\
					\hline
					Control Sequence\cite{Tan}                &     $MZ(M-1)$                                  &    $M$                                                                &    $MZ$                                                   &    $M$                                                                      &     \ding{51}                               &     \ding{51}                                                         &     \ding{55}\\
					\hline
					This paper(Con.1)                         &     $tR$                                       &    $\sum_{i=0}^{\lambda-1}u_i$                                        &    $t$                                                    &    $\sum^{\lambda-1}_{i=0}u_i$                                              &     \ding{51}                               &     \ding{51}                                                         &     \ding{51}\\
					\hline
					This paper(Con.2)                         &     $d^nt$                                     &    $\sum^{\lambda-1}_{i=0}u_i$                                        &    $t$                                                    &    $\sum^{\lambda-1}_{i=0}u_i$                                              &     \ding{51}                               &     \ding{51}                                                         &     \ding{51}\\
					\hline
				\end{tabular}
			}
		\end{center}
	\end{table*}
	Combining the information provided in the table and the preceding descriptions in Section~\ref{intro}, it can be observed that protocol sequences and frequency hopping sequences are not suitable for data link communication with TDMA architecture. Furthermore, both of them allow for collisions and are not suitable for collision-free communication among users or hierarchical access communication. In contrast, HCSs build upon existing control sequences, and enables the collision-free hierarchical access communication within the context of a data link time-division multiple access architecture.
	\subsection{Simulation Results}
	We present the simulation results based on the hierarchical access control schemes and the fixed time slot assignment scheme under different interference numbers and powers, respectively.
	Since the proposed HCSs concern slot utilization and anti-interference performance, users will benefit from the full slot utilization as well as the improved interference immunity.
	Assume that all users are communicating via the Additive Gaussian White Noise(AGWN) channel and that each user accesses one channel.
	We simulate the symbol error ratio(SER) under the following conditions:
	$(a)$ the number of time slots is $t=8$;
	$(b)$ the access level of users is $r=4$;
	$(c)$ the lengths of the HCSs generated by the above two constructions are 32 and 128, respectively.
	$(d)$ the time slots in the fixed time slot assignment scheme are [0,2,4,5];
	$(e)$ the interference point for single-interference case is set in time slot[2];
	$(f)$ the interference point for multiple-interference case is set in time slots [1,4,5];
	$(g)$ the interference powers are 10dB and 15dB respectively.
	\begin{figure}[htb]
		\begin{center}
			\subfloat[\label{single_dB10} Interference power of 10dB]{\includegraphics[width=0.7\columnwidth]{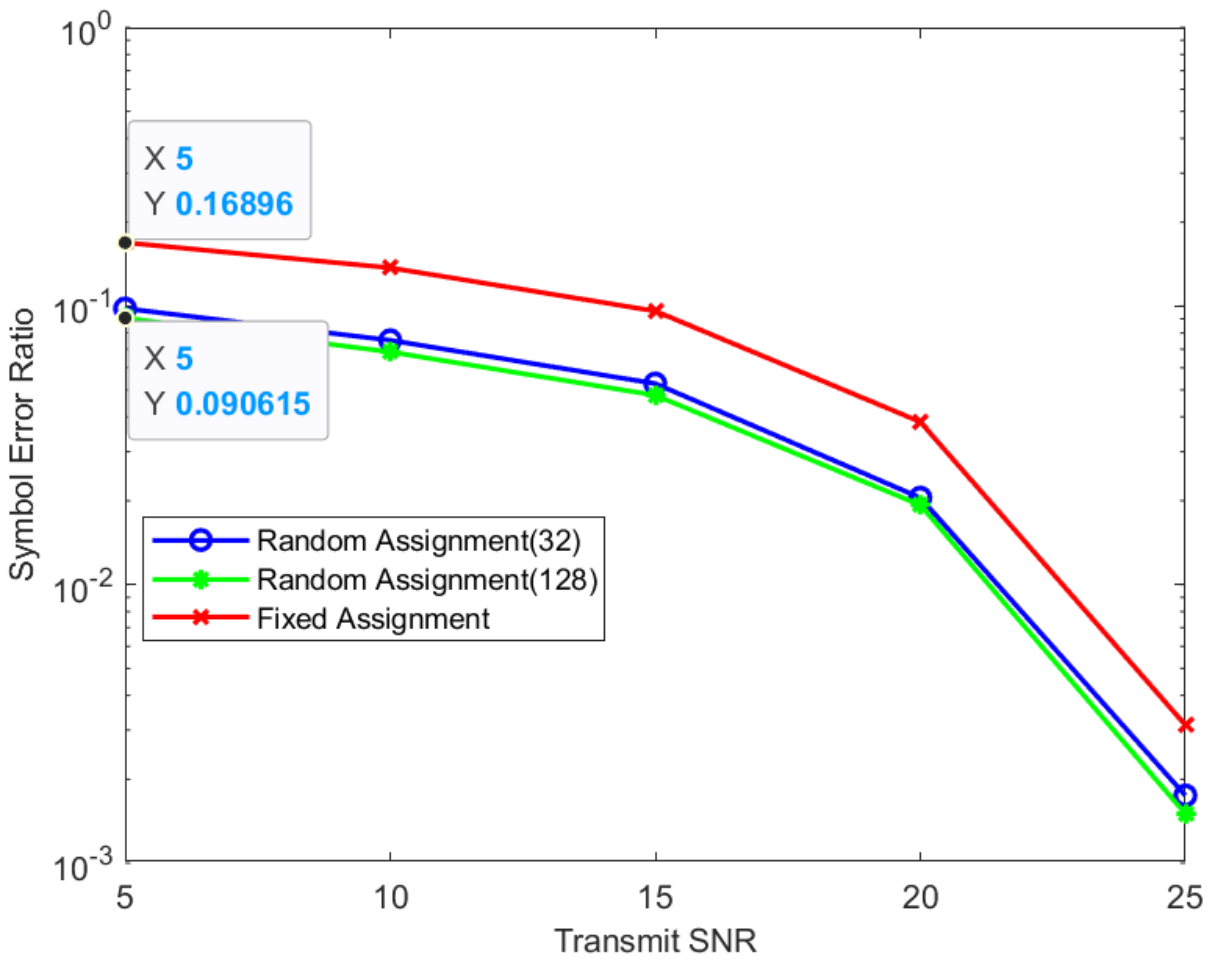}}
			\hfill
			\subfloat[\label{single_dB15} Interference power of 15dB]{\includegraphics[width=0.7\columnwidth]{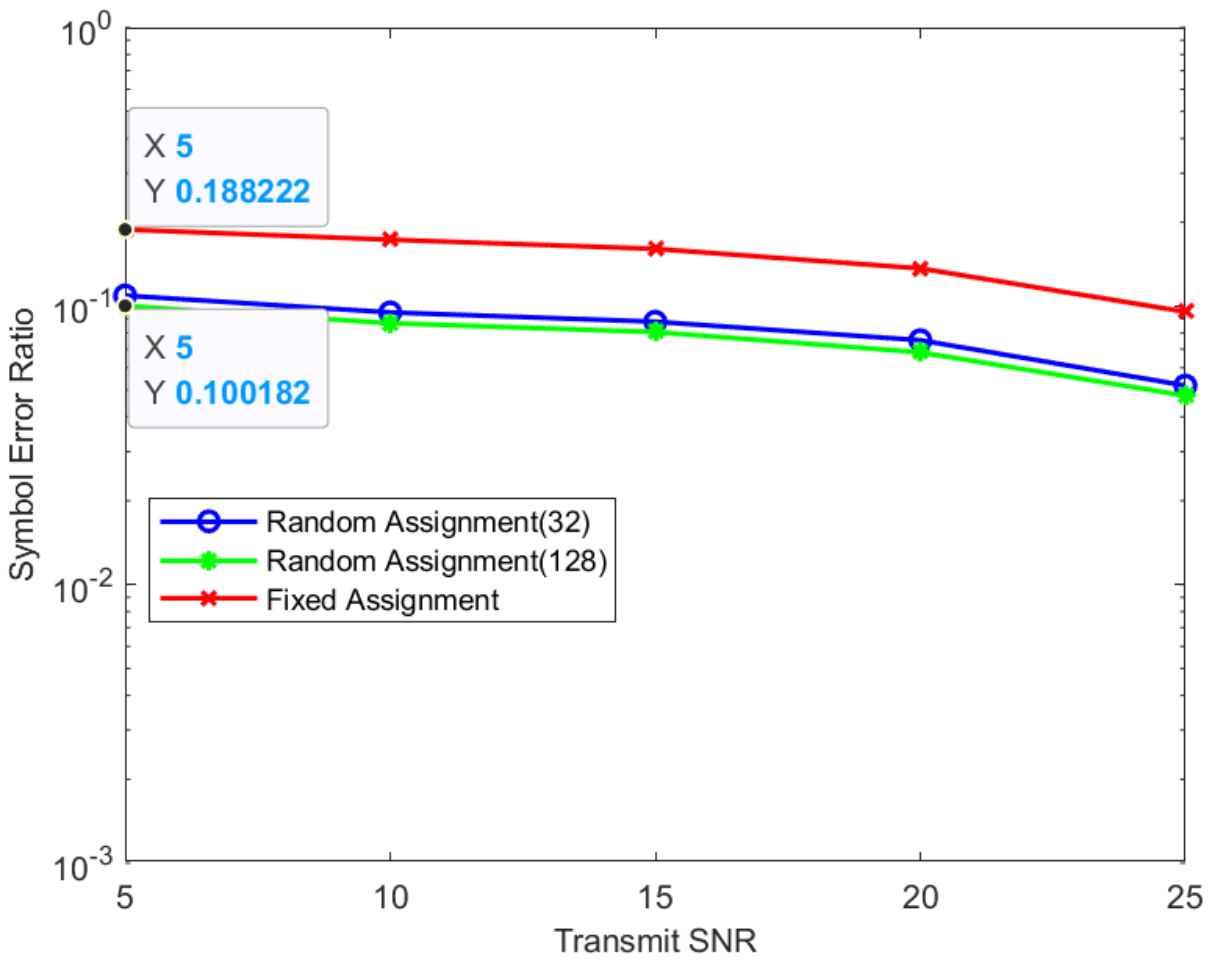}}
		\end{center}
		\caption{The SER under the case of single-interference}
		\label{fig:single}
	\end{figure}
	\begin{figure}[htb]
		\begin{center}
			\subfloat[\label{multi_dB10} Interference power of 10dB]{\includegraphics[width=0.7\columnwidth]{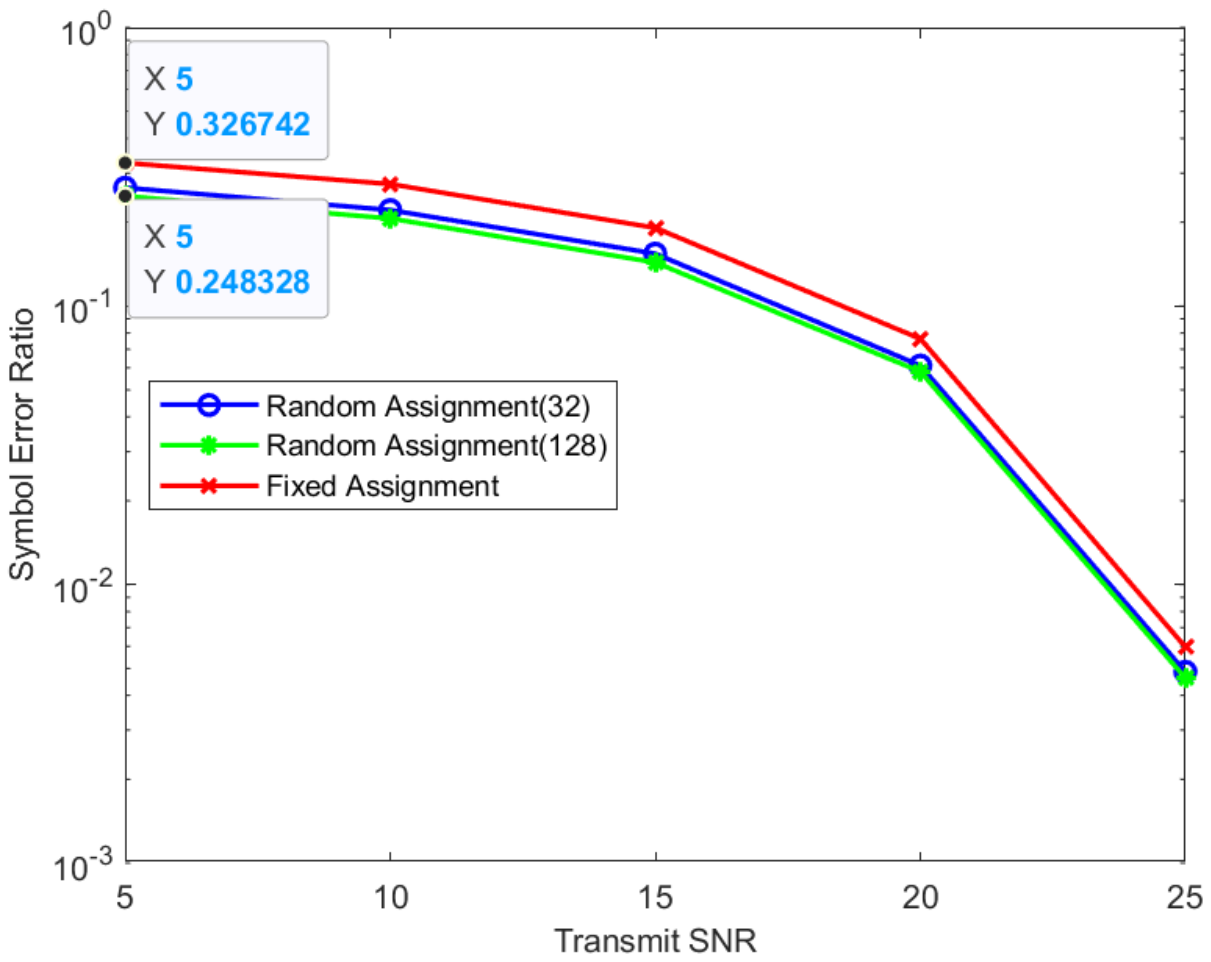}}
			\hfill
			\subfloat[\label{multi_dB15} Interference power of 15dB]{\includegraphics[width=0.7\columnwidth]{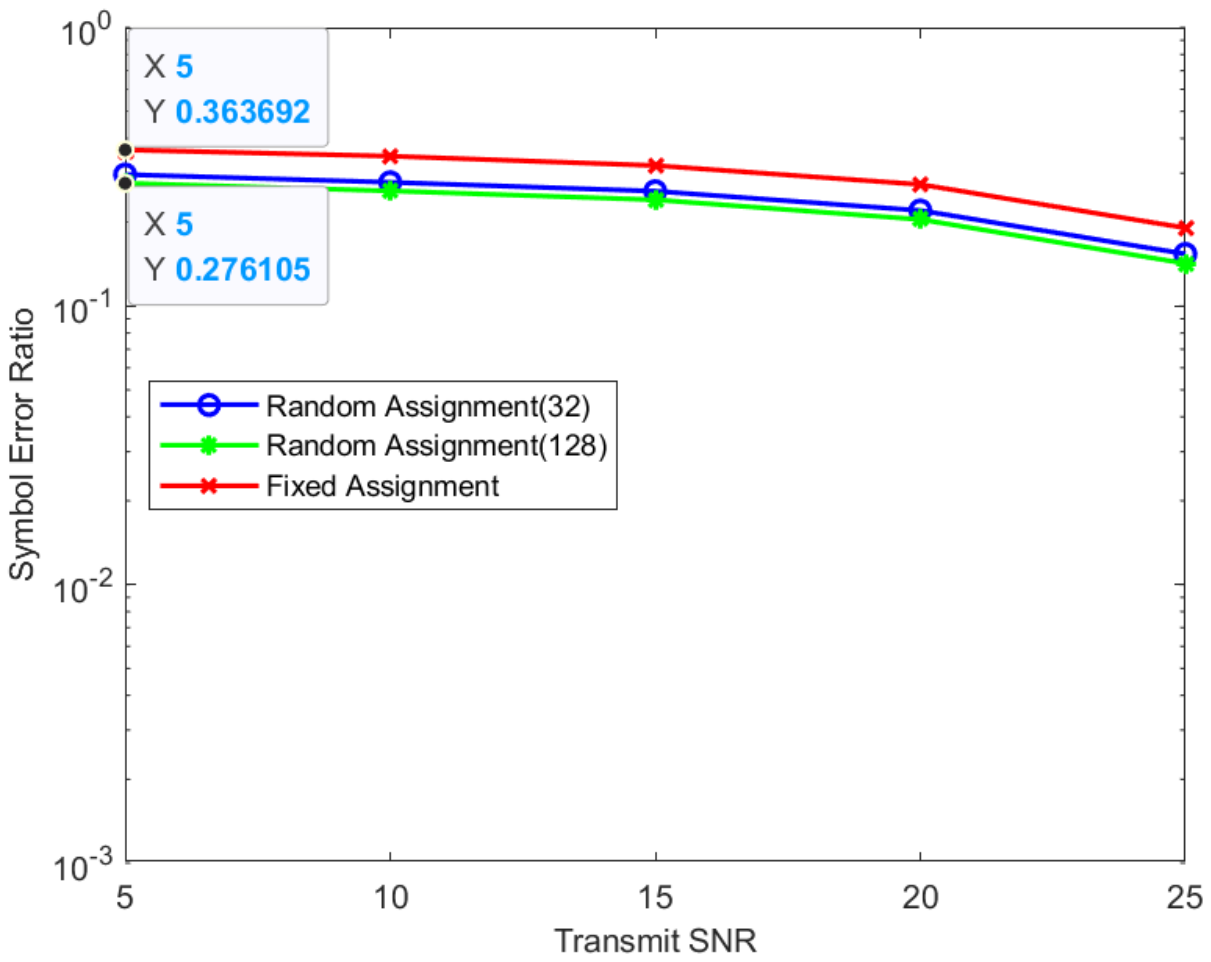}}
		\end{center}
		\caption{The SER under the case of multiple-interference}
		\label{fig:multi}
	\end{figure}
	First of all, we can see from Fig.\ref{fig:single} that the SER of the HCS-based hierarchical access control scheme is lower than that of the fixed time slot assignment scheme.
	Furthermore, when the SNR is 5dB, the maximum difference in SER between the two schemes is the greatest(about 7$\%$), which means that in the case of single-interference and interference power of 10dB, the anti-interference performance of the hierarchical access control scheme can be improved by up to about 7$\%$ compared to the fixed time slot assignment scheme.
	And in the case of interference power of 15dB, the anti-interference performance of the hierarchical access control scheme can be improved by up to about 8$\%$.
	In a multiple-interference environment(Fig.\ref{fig:multi}), compared to the fixed time-slot assignment scheme, the hierarchical access control scheme has the maximum anti-interference improvement of approximately 7$\%$ and 8$\%$ at different interference powers.
	Additionally, we compare the SER of different lengths of sequences generated by the above two constructions under the same parameters of $t=8,r=4$. As can be seen from these figures, on the basis of better anti-interference performance than the fixed time slot assignment scheme, the longer the HCS, the better its anti-interference performance.
	Based on this simulation result, we can conclude that longer HCSs can be used for sending and receiving long messages between users, while shorter HCSs are more suitable for transmitting short messages.
	\section{Conclusions}\label{Conclu}
	In this paper, we analyzed the significance of hierarchical access control in data link networks, and proposed a HCS-based hierarchical access control model to address this problem.
	Then, in order to better measure the performance of HCSs, we derived a new theoretical bound for HCS in terms of access levels of users, number of users and number of time slots.
	Based on the new bound, two classes of optimal HCS sets are constructed. Using the HCSs, users with different access levels can achieve hierarchical, random, and secure access in data link for communication, and ensure maximum utilization of time slots without conflicts.
	Finally, through the performance evaluation, it is shows that HCSs is the optimal choice for hierarchical access communication among users in data link with TDMA architecture. Moreover, the excellent anti-interference performance of the HCS-based hierarchical access control scheme is demonstrated, and the longer the HCS, the better the anti-interference performance it exhibits.
	It is expected that the new bound and new constructions will be useful for future research on the problems of hierarchical access control in data link networks.
	\section*{Acknowledgment}
	The work of this paper was supported by the National Science Foundation of China (No. 62171387), the China Postdoctoral Science Foundation (No. 2019M663475).
	\bibliographystyle{unsrt}
	\bibliography{REF}
 
\end{document}